\documentclass[a4paper,UKenglish]{lipics}

\usepackage{microtype}
\usepackage{algorithm}
\usepackage{algorithmic}
\usepackage{cite}

\newtheorem{thm}{Theorem}[section] 
\newtheorem{lem}[thm]{Lemma}

\newtheorem{clm}[thm]{Claim}

\newtheorem{rmk}{Remark}
\newtheorem{cjt}{Conjecture}

\theoremstyle{definition}
\newtheorem{defn}[thm]{Definition}

\newcommand{\Maj}{\text{Maj}}
\newcommand\E{\mathbb{E}}

\newcommand\I{\mathbb{I}}

\newcommand{\hide}[1]{} 
\newcommand{\diff}{\text{diff}}

\title{Information Cascades on Arbitrary Topologies}
\titlerunning{Information Cascades on Arbitrary Topologies} 

\author[1]{Jun Wan}
\author[1]{Yu Xia\footnote{Part of this work was done when this author visited Microsoft Research Asia.}}
\author[2]{Liang Li}
\author[2]{Thomas Moscibroda}
\affil[1]{The Institute for Theoretical Computer Science(ITCS), Institute for Interdisciplinary Information Sciences, Tsinghua University\footnote{This work was supported in part by the National Basic Research Program of China Grant 2011CBA00300, 2011CBA00301, the National Natural Science Foundation of China Grant 61361136003.}\\
  \texttt{\{wanj12, xiay12\}@mails.tsinghua.edu.cn}}
\affil[2]{Microsoft Research\\
  \texttt{\{liangl, moscitho\}@microsoft.com}}
\authorrunning{J. Wan, Y. Xia, L. Li and T. Moscibroda} 

\Copyright{Jun Wan, Yu Xia, Liang Li and Thomas Moscibroda}

\subjclass{F.2.2 Nonnumerical Algorithms and Problems, G.2.2 Graph Theory}
\keywords{Information Cascades, Herding Effect, Random Graphs}

\begin{document}

\maketitle

\begin{abstract}
In this paper, we study information cascades on graphs.
In this setting, each node in the graph represents a person.
One after another, each person has to take a decision based on a private signal as well as the decisions made by earlier neighboring nodes.
Such information cascades commonly occur in practice and have been studied in complete graphs where everyone can overhear the decisions of every other player.
It is known that information cascades can be fragile and based on very little information, and that they have a high likelihood of being wrong.

Generalizing the problem to arbitrary graphs reveals interesting insights.
In particular, we show that in a random graph $G(n,q)$, for the right value of $q$, the number of nodes making a wrong decision is logarithmic in $n$.
That is, in the limit for large $n$, the fraction of players that make a wrong decision tends to zero.
This is intriguing because it contrasts to the two natural corner cases: empty graph (everyone decides independently based on his private signal) and complete graph (all decisions are heard by all nodes). In both of these cases a constant fraction of nodes make a wrong decision in expectation. Thus, our result shows that while both too little and too much information sharing causes nodes to take wrong decisions, for exactly the right amount of information sharing, asymptotically everyone can be right.
We further show that this result in random graphs is asymptotically optimal for any topology, even if nodes follow a globally optimal algorithmic strategy. Based on the analysis of random graphs, we explore how topology impacts global performance and construct an optimal deterministic topology among layer graphs.
\end{abstract}

\section{Introduction}\label{sec:intro}

An Information Cascade occurs when a person observes the actions of others and then
	---~in spite of possible contradictions to his/her own private information
	---~follows these same actions.
A cascade develops when people
	``abandon their own information in favor of inferences based on earlier people's actions''\cite{easley2010networks}.
Information Cascades frequently occur in everyday life.
Commonly cited examples include the choice of restaurants when being in an unknown place
	people choose the restaurant that already has many guests over a comparatively empty restaurant,
or hiring interview loops where interviewers follow earlier interviewer's decisions if they are not sure about the candidate.
Notice that information cascades are not irrational behavior;
	on the contrary,
	they occur precisely because people rationally decide based on inferences derived from earlier people's actions.

The simple herding experiment by Anderson and Holt illustrates Information Cascades \cite{anderson1996classroom, anderson1997information}(see also Chapter 16 in \cite{easley2010networks}).
In this experiment, an urn contains three marbles,
	either two red and one blue (majority red),
	or one red and two blue (majority blue).
The players do not know whether the urn is majority red or blue.
One by one, the players privately pick one marble from the urn,
	check its color, return it to the urn,
	and then publicly announce their guess as to whether the urn is majority red or majority blue.
The first and second player will naturally base their guesses on the colors of the marble they picked,
	thus their guesses reveals their private signals.
For any subsequent player however,
	her rational guess may not reflect her own signal.
For example, suppose the first two players both guess \emph{red}. In this case,
	it is rational for the third player to also guess \emph{majority-red} regardless of the color of the marble she picked.
Indeed, the third player makes her decision on a rational inference based on the first two guesses.
Since her guess does therefore not reveal any further information about the urn to any subsequent player,
	\emph{every} subsequent player will guess the urn to be \emph{majority-red}.
The example shows that information cascades can be based on very little actual information and thus fragile;
	and they can be wrong.
Indeed, in the above example with urns,
	it can be shown that with probability $1/5$, a ``wrong cascade'' occurs,
	i.e., all players (except from possibly a few at the beginning) will guess wrongly.

The standard model for information cascades studies the process in which
	players make decisions sequentially based on their own private signals
	as well as the set of decisions made by earlier players~\cite{banerjee1992simple, bikhchandani1992theory, welch1992sequential}.
In this paper, we interpret and generalize the traditional information cascade setting as a game in a graph.
Each player is a node,
	and an edge between two nodes $v$ and $w$ means that $w$ can hear about $v$'s guess
	(assuming $w$ is after $v$ in the order of decision-making).
Thus the traditional information cascade model corresponds to a complete graph
	(all players hear the decisions of all other players).
At the other end of the spectrum,
	the empty graph means that every player decides independently of all other players, purely based on their own private signal.
Casting the information cascade problem in this graph setting
	allows us to study the range in between the two extreme points of complete and empty graphs.

\begin{figure}
	\begin{minipage}{2in}
    		\includegraphics[scale = 0.55]{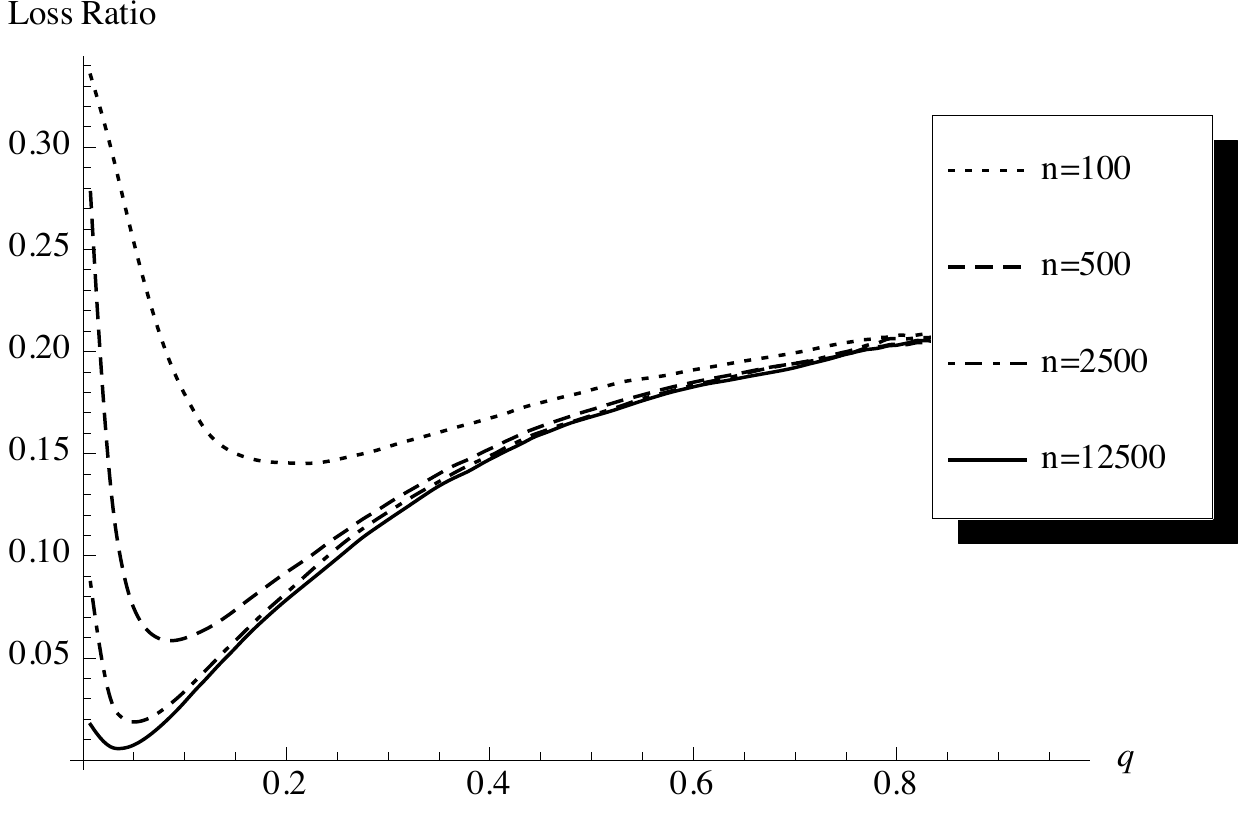}
		\caption{Performance of random graphs for different $q$ and $n$.}
		\label{fig:random}
	\end{minipage}
	\hspace{1in}
	\begin{minipage}{2in}
		\includegraphics[scale = 0.55]{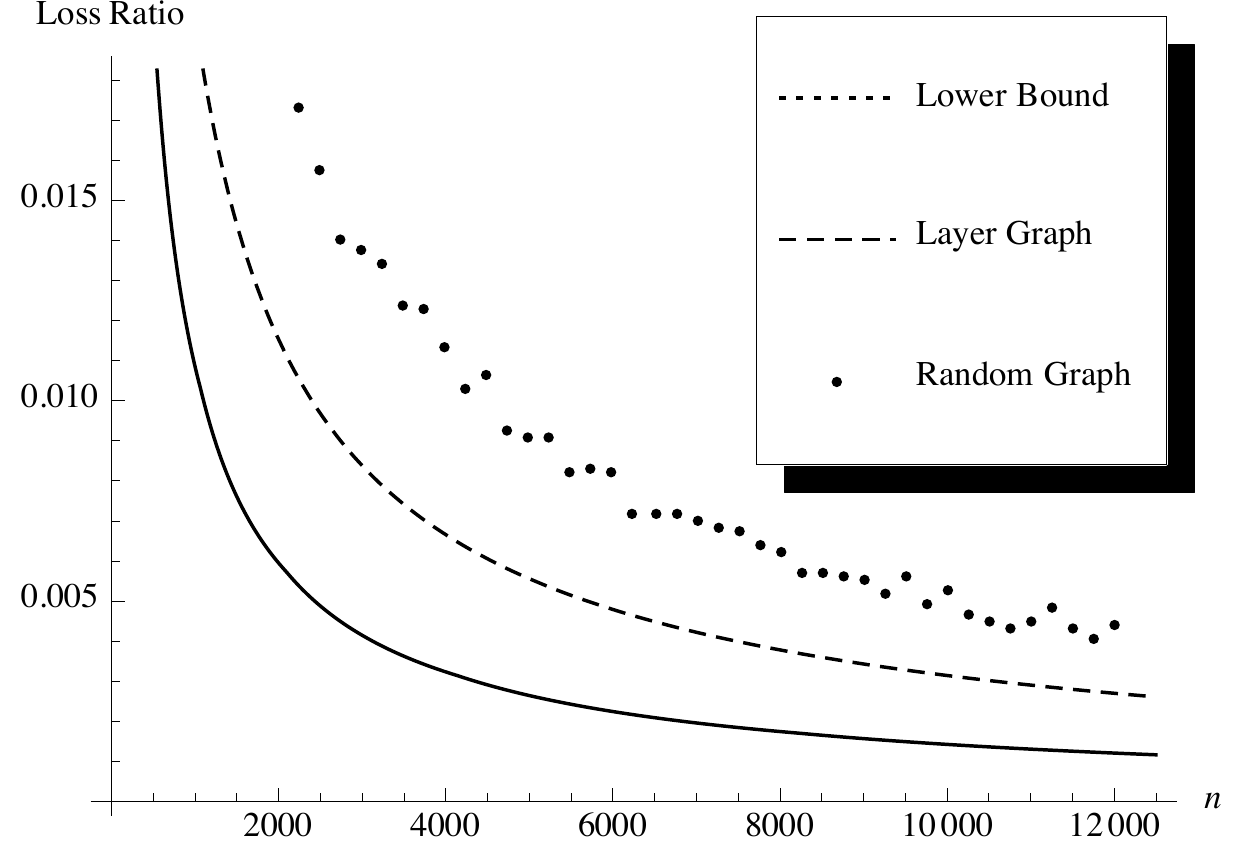}
		\caption{Performance of different topologies and strategies.}
		\label{fig:simulation}
	\end{minipage}
\end{figure}

Studying this range in between reveals fascinating insights.
Figure \ref{fig:random} shows the expected number of wrong guesses
	in the above 3-marble-urn experiment in a random graph $G(n,q)$ topology,
	for different values of $n$ and $q$.
In the empty graph ($q=0$), if all nodes take their decisions independently, $1/3$ of the players are wrong.
In the complete graph ($q=1$), $1/5$ of the players are wrong on average as discussed above.
However, the interesting thing is that for some values in between these two extremes,
	the number of wrong decisions is significantly less.
Indeed, it seems that for the right value of $q$ and $n\rightarrow\infty$, the number of wrong decisions tends to $0$.

These observations are intriguing:
	It looks like that if people share too much information,
	a constant fraction of the population is wrong because of bad information cascades occurring.
If people share too little information,
	a larger constant fraction of the population is wrong because the players take their decisions too independently, relying too much on their private signal which has a constant probability of being wrong.
But, if exactly the right amount of information is shared,
	then it seems that in the limit, \emph{all players} (at least asymptotically) take the correct decision.

In this paper, we study this phenomenon.
We prove that, indeed, in a random graph the number of wrong nodes is at most $O(\log n)$ for the optimal value of $q$ (Section 3).
We then study arbitrary graph topologies and show that $O(\log n)$ wrong nodes is optimal in a strong sense (Section 4).
Specifically, even in the best possible topology, there are at least $\Omega(\log n)$ wrong nodes.
This result holds even if a global oracle tells each node whether it should
	a) base its decision solely on its private signal (thus revealing this signal as additional information to all its neighbors) or
	b) base its decision on the majority of private signal and neighboring decisions as in the cascade model above.
In other words, even if nodes can ``sacrifice'' themselves to reveal additional information to their neighbors
	and even in the best possible topology $\Omega(\log n)$ wrong nodes is a lower bound.
Finally, we derive an optimal deterministic topology from among a family of \emph{layer graphs} (Section 5).

\section{Related Work}
Sequential decision-making has been studied in various areas including politics, economics and computer science\cite{Arthur1989, banerjee1992simple, bikhchandani1992theory,Granovetter1978,easley2010networks}. The primary concern on the Bayesian learning model\cite{bikhchandani1992theory, Smith99pathologicaloutcomes, banerjee1992simple, welch1992sequential, Smith2008rational, word_of_mouth2004, BayesianLearning} is under what conditions asymptotically correct information cascades occur. For specific graph topologies such as complete graphs and line graphs, conditions on the private signals were addressed to guarantee the correctness of cascades\cite{Smith99pathologicaloutcomes, observationalLearning}. For arbitrary graph topologies, the approach of Acemoglu et al.\cite{BayesianLearning} is intuitively quite consistent to our $k$-layer topology(see Section \ref{sec:k layer}) and can be used to explain why our random network and selfishless decision-making algorithm achieve global optimality. While their approach focuses on the asymptotic probability of correct cascades, our result can quantitatively bound the expectation number of incorrect nodes.

There has been research on different sequential decision making models in graphs. For example, Chierichetti et al. \cite{Chierichetti2012} study different algorithms for finding appropriate orderings to maximize the fraction making correct decisions and Hajiaghayi et al.\cite{Hajiaghayi2013} and Hajiaghayi et al.\cite{Hajiaghayi2014} generalizes the model and improve related bounds. However, notice that the threshold decision-making processes studied in these works fundamentally differ from the information cascade setting we consider in this work. Indeed, the effect of too little/too much information sharing being bad as shown in Figure \ref{fig:random} is not observed in such threshold models.


%
There also exists an impressive body of work on sequential and non-sequential decision on arbitrary graphs that however do not capture information cascades as exemplified in Anderson and Holt's herding experiment. Typically, each node updates its opinion through repeated averaging with neighbors. General conditions for convergence to consensus have been developed\cite{Acemoglu2010, DeGroot, Golub07naivelearning}. Intrigued by the observation that consensus is usually not reached  in real world\cite{Krackhardt47}, Bindal et al.\cite{Bindel11} use a game theoretic approach to study the equilibrium of the dynamical process and measure the cost of disagreement via the Price of Anarchy\cite{Koutsoupias1999}.

\section{Preliminaries}\label{sec:prel}

We introduce the formal definitions of our model.
There are $n$ nodes (numbered $1,2,\cdots n$) whose neighboring relationship is depicted by a graph $G = ([n],E)$.
All nodes make decisions sequentially according to their numbers in order to guess a global ground truth value $b\in\{0,1\}$.
When making its decision, each node can only obtain a random partial information on $b$.
That is, when node $i$ observes $b$,
	it can only get a \emph{private signal} $s_i$ which equals $b$ with probability $p>0.5$ or equals $1-b$ with probability $1-p$.
The decision-making of a node not only depends on its private signal observed from $b$,
	but also on the decisions made by its previous neighbors. Note that the neighboring decisions may or may not be based on those nodes' private signals.
More formally, let $c_i$ be the \emph{output decision} or \emph{guess} of node $i$ and $c^i$ be the \emph{decision vector} $(c_1, c_2, \cdots, c_i)$,
	if $ L_i : \{0,1\}^{i-1}\times \{0,1\} \to \{0,1\} $ is the \emph{decision-making algorithm} for node $i$,  we have in general $c_i = L_i(c^{i-1}, s_i)$.

Given the graph $G$ and decision-making algorithms $L_1, L_2, \cdots, L_n$,
	we use $\mathcal{E}_G(L_1, \dotsb, L_n)$ to denote the expected number of nodes that output the wrong value $1 - b$.
When it is clear from the context, we may abbreviate this notation to $\mathcal{E}_G$,
	$\mathcal{E}(L_1, \dotsb, L_n)$, or simply $\mathcal{E}$.
The global objective of this sequentially decision-making process is to minimize $\mathcal{E}_G(L_1, \dotsb, L_n)$,
	which is equivalent to maximizing the expected number of nodes who guess the ground truth value correctly.
We will show that such an optimization task can be achieved by adjusting the graph topology or the decision-making algorithms.

Let $\mathcal{E}_i$ be the \emph{failure probability} that node $i$ outputs $1 - b$.
As a node often makes inferences based on others' decisions without knowing their private signals,
	it is intuitively understandable that a node's probability of correct decision-making can be quantified by the number of private signals it can infer.

In reality, the Majority Algorithm is one of the most popular and practical algorithms for decision-makings.
This kind of ``following the herd'' algorithm can often achieve a locally optimal effect. In this paper, we use $\Maj_k$ to denote the Majority Algorithm taking input bits of length $k$. We just use $\Maj$ if $k$ is clear from the context.
In Chapter 16 of \cite{easley2010networks}, Easley and Kleinberg shows that the Majority Algorithm is optimal when a node observes multiple independent signals.
\begin{clm}
	\label{clm:multiple is optimal}
	For any node $i$ seeking to maximize $\mathcal{E}_i$,
		when it observes multiple signals(including its own private signal), its optimal algorithm is to output the majority of these observed signals.
\end{clm}

However, Anderson and Holt's experiment shows that if all nodes apply the Majority Algorithm,
	it is possible that essentially all of the nodes guess incorrectly,
	leading to an information cascade on the wrong side.
In this paper, we address this problem and analyze the impact of topology and algorithms on information cascades.

\section{Random Graphs}
In this section, we analyze the performance of the Majority Algorithm on random graphs.
Conventionally, $G(n, q)$ denotes the random graph model that generates a random graph with $n$ nodes and each pair of nodes are connected by an edge with probability $q$.
Different connection probabilities can result in completely different topologies, and thus dramatic changes in $\mathcal{E}_{G(n,q)}$.
As introduced in Section \ref{sec:intro}, when $q$ equals 0 or 1 corresponding to the empty or complete topology, the expected number of wrong output decisions are both $\Theta(n)$.

In this section, we show that there exists a  $q$ such that the Majority Algorithm can achieve only $\Theta(\log n)$ expected wrong output decisions:

\begin{thm}
\label{thm:random graph log n reachable}
There exists a connection probability $q = \Theta(1/\log n)$ such that in $G(n, q)$, when all nodes apply the Majority Algorithm, we have $\mathcal{E}_{G(n,q)} = \Theta(\log n)$.
\end{thm}

	

We can also demonstrate the optimality of this connection probability by showing a lower bound for the expected number of wrong decisions:

\begin{thm}
\label{thm:random graph log n bound}
For any connection probability $q$, when all nodes apply the Majority Algorithm, the expected number of wrong outputs in $G(n,q)$ is lower bounded by $\Theta(\log n)$, i.e. $\mathcal{E}_{G(n,q)} = \Omega(\log n)$.
\end{thm}

A key ingredient to our proof is to bound the failure probability $\mathcal{E}_i$ for each node. Applying the Chernoff Bound and the Union Bound, we can further bound the overall $\mathcal{E}_{G(n,q)}$. The following are two technical lemmas for bounding the failure probabilities(proofs in Appendix \ref{subsec:random exponential bound} and Appendix \ref{subsec:random general bound}):

\begin{lem}
\label {lem:random exponential bound}
If a constant fraction $f > 0.5$ of the first $i$ nodes are correct (resp. wrong),
	node $i+1$'s failure (resp. correct) probability $\mathcal{E}_{i+1}$ is upper bounded by $e^{-\Theta(iq)}$.
\end{lem}

%

\begin{lem}
\label {lem:random general bound}
If a constant fraction $f > 0.5$ of the first $i$ nodes are correct, s.t.
$
	{f\over 1-f} \ge \sqrt{q\over 1-q},
$
then node $i+1$'s failure probability is upper bounded by $p$.
\end{lem}


\subsection {Proof of Theorem \ref{thm:random graph log n reachable}}

In this subsection, we provide a detailed proof of Theorem \ref{thm:random graph log n reachable}.

The reason why random graphs behave well for the right value of q is that randomness defers the process of information cascades.
The fewer neighbors, the more likely a node will output its private signal, thereby 1) having a high probability of being wrong, but 2) revealing important information to its neighbors. 
When $q = \Theta(1/ \log n)$, with high probability each of the first $\Theta(\log n)$ nodes can have at most one neighbor.
By definition of Majority Algorithm, any node with only one neighbor will be forced to output its own private signal.


Using Lemma \ref{lem:random exponential bound}, we can prove that an established cascade among the first $\log n / q$ nodes decides the outputs of all later nodes with high probability(proof in Appendix \ref{subsec:first logn/q deciding}):
\begin{lem} \label {lem:first logn/q deciding}
If among the first $\log n/q$ nodes,
	only a small constant fraction $f < 0.5$ output wrongly,
	then for the later $n - (\log n/q)$ nodes, the expected number of wrong outputs is at most $O(1)$.
\end{lem}


Lemma \ref{lem:first logn/q deciding} is insufficient to bound the $\Theta(\log n)$ expected failure nodes as required by Theorem \ref{thm:random graph log n reachable},
	in that it only bounds the loss of later nodes in the sequence. It could be the case that the first $\log n/q = \Theta(\log^2 n))$ nodes all fail.
To bound the overall $\mathcal{E}$, it is essential to analyze the performance of the first $\log n/q$ nodes.
We can use an induction argument to show that for the optimal $q$, the first $\log n/q$ nodes are majority-correct with high probability:


\begin{lem}
\label{lem:majority among first logn/q}
Let $\delta = {1\over 2}(p + {\sqrt{p}\over \sqrt{p} + \sqrt{1-p}})$.
There exists connection probability $q_{opt} = \Theta(1/\log n)$ such that the first $\Theta(\log n/q)$ nodes
	contains at least $\delta$ portion of correct outputs with probability $1 - O(n^{-1}\log n)$.
\end{lem}

\begin{proof}
We can prove Lemma \ref{lem:majority among first logn/q} by induction.
Consider dividing the first $\Theta(\log n/q)$ nodes into $\Theta(\log n)$ segments,
	where each segment contains $\Theta(1/q) = \Theta(\log n)$ many nodes.
We analyze each segment independently and show that
\begin{itemize}
\item There exists $q_1 = \Theta(1/\log n)$,
		such that the first segment contains $\delta$ portion of correct outputs with probability at least $1 - O(n^{-1})$.
\item If the first $i$ segments contain $\delta$ portion of correct outputs,
		then the $(i+1)^{th}$ segment will also contain $\delta$ portion of correct outputs with probability $1 - O(n^{-1})$.
\end{itemize}
By a single Union Bound, we can combine these two results and show that
the first $\Theta(\log n/q)$ nodes contain $\delta$ portion of correct outputs with probability $1 - \log n\cdot O(n^{-1})$.
The detailed proof is provided in Appendix \ref{subsec:majority among first logn/q}.
\end{proof}


 Lemma \ref{lem:majority among first logn/q} and \ref{lem:random exponential bound} together imply a $\Theta(\log n)$ upper bound for the expected number of wrong outputs among the first $\Theta(\log n / q)$ nodes(proof in Appendix \ref{subsec:loss of former nodes}):
\begin{lem}
\label{lem:loss of former nodes}
If $q = q_{opt}$, the first $\Theta(\log n/q)$ nodes' expect to have at most $\Theta(\log n)$ many wrong outputs.
\end{lem}

\begin{proof}[Proof of Theorem \ref{thm:random graph log n reachable}]
With Lemma \ref{lem:first logn/q deciding}, \ref{lem:majority among first logn/q} and \ref{lem:loss of former nodes} proved, the expected number of wrong output decisions under connection probability $q_{opt}$ is bounded by
\begin{equation}
\begin{aligned}
\label{equ:bound2}
	\mathcal{E}_{G(n,q)}
& =
	\sum_{i=1}^{n}	\mathcal{E}_i
	=
	\sum_{i=1}^{\log n/q}	\mathcal{E}_i + \sum_{i=\log n/q + 1}^{n}	\mathcal{E}_i
\\
& \le
	\Theta(\log n) + (1 - O(n^{-1}\log n))\cdot O(1) + O(n^{-1}\log n)\cdot n
	=
	\Theta(\log n).
\end{aligned}
\end{equation}
which completes the proof.
\end{proof}

\subsection {Proof of Theorem \ref{thm:random graph log n bound}}
In the previous section, we prove that for the optimal connection probability $q_{opt} = \Theta(1 / \log n)$, the expected number of wrong outputs can be reduced to $\Theta(\log n)$.
However, it remains a problem whether we can move beyond $\Theta(\log n)$. In this section, we prove Theorem \ref{thm:random graph log n bound} which states that the bound in Theorem \ref{thm:random graph log n reachable} is asymptotically optimal.

\begin{proof}[Proof of Theorem \ref{thm:random graph log n bound}]

We prove this theorem for two separate cases, namely when $q = O(1 / \log n)$ and $q = \omega(1 / \log n)$.

When $q = O(1 / \log n)$, the intuition is that we need at least $\Theta(1/q)$ nodes before accumulating an actual influential cascade.
For the $i^{th}$ nodes where $i \le 1/q$, its chance of being isolated is
$
	(1 - q)^{i} + iq(1-q)^{i-1}
\ge
	(1 - q)^{1 / q}
\sim
	{1 / e}.
$
Therefore the node's failure probability is at least $\mathcal{E}_i = (1-p)\cdot\Pr[\text{isolated}] = (1-p)/e$.
This lower bounds the expected number of failure nodes by $(1-p)/(eq) = \Theta(1/q)$.

When $q = \omega(1 / \log n)$,
	a wrong cascade occurs with high probability,
	thus resulting in a significant number of failure nodes.
With probability $(1-p)^{\Theta(1/q)}$, all of the first $\Theta(1/q)$ nodes observe a wrong signal and output the wrong guesses.
Using Lemma \ref{lem:random exponential bound}, we can show that with high probability, the majority of later nodes follow this wrong cascade.
Therefore, the total number of failure nodes is at least
$
(1-p)^{\Theta(1/q)}\cdot \Theta(n) = n^{1 - o(1)} = \Omega(\log n).
$
\end{proof}

\section{General Lower Bound}
In this section, we design a non-constructive scheme that finds the optimal decision-making algorithms for general graphs.
Given the neighboring graph $G = ([n], E)$, our goal is to find the set of algorithms $\{L_i\}_{i=1}^{n}$ such that
$
	(L_1, \dotsb, L_n)
=
	\arg\min_{L'_1, \dotsb, L'_n} \mathcal{E}_{G}(L'_1,\dotsb, L'_n).
$

An important use of the non-constructive scheme is to provide a general lower bound for arbitrary topology.
For any set of decision-making algorithms in a topology $G$, we can simulate it on a complete graph by considering only edges in $G$.
Thus the minimal $\mathcal{E}$ for complete graphs is a general lower bound for arbitrary topology:

\begin{thm}
\label{thm:complete graphs bound}
The expected number of wrong nodes $\mathcal{E}$ under the optimal decision-making algorithms of complete graphs lower bounds the $\mathcal{E}$ of any algorithms in any topology.
\end{thm}

From our previous discussion on random graphs, we know that the expected number of wrong guesses $\mathcal{E}$ highly depends on the number of nodes revealing their private signals.
This inspires us to make the following definitions:
\begin{defn}
Node $i$ reveals \textbf{valid} information under $c^{i-1}$
	if and only if node $i$ outputs its private signal under $c^{i-1}$,
	i.e. $c_i = L_i(c^{i-1}, s_i) = s_i$.
Furthermore, we denote $\text{Valid}(\cdot)$ as a function that extracts a vector of valid information out of a decision vector, i.e.
	$c_j$ is in the vector $\text{Valid}(c^{i})$ if and only if $c_j$ is valid.
\end{defn}

\begin{defn}
A node $i$'s \textbf{reveal set} $RS_i$ is the set of $c^{i-1}$ which causes node $i$ to reveal \textbf{valid} information.
\end{defn}

Note that any valid information is correct with probability $p$ and is independent of other nodes. Using the same Bayesian argument\cite{easley2010networks}, we can prove a similar lemma as Claim \ref{clm:multiple is optimal} in Section \ref{sec:prel}, which states that a node's guess is beneficial for later nodes if and only if the guess is valid(proof in Appendix \ref{subsec:selfMaj}):

\begin{lem}
\label{lem:selfMaj}
For a node $i$ seeking to minimize its failure probability $\mathcal{E}_i$, the optimal decision-making algorithm is to perform the Majority Algorithm on $\text{Valid}(c^{i - 1})\bigcup\{s_i\}$, i.e. $c_i = \Maj(\text{Valid}(c^{i - 1}), s_i)$.
\end{lem}


\subsection{A non-constructive optimal algorithm scheme for general graphs}
\label{Construction}
In this section, we provide a general scheme for finding the optimal decision-making algorithms of all nodes in arbitrary topologies. Our scheme is non-constructive in that it neither explicitly specifies what the optimal algorithms are, nor shows how to find them efficiently.



Given the underlying topology, all the nodes decide their algorithms sequentially in a \emph{greedy} way as follows.
Node $1$ publicly announces $L_1$, based on its own rationality, then node 2 announces $L_2$ with the knowledge of $L_1$, etc(see Algorithm \ref{algo:non-const}).
Any node $i$ will base its knowledge on $L_1, \dotsb, L_{i-1}$ when deciding $L_i$.
Each node designs its own decision-making algorithm in order to locally minimizes the failure probability.
Denote this construction scheme as $GC(\cdot)$, the abbreviation of ``greedy construction'', then for node $i$, we have $L_i = GC(L_1,\dotsb,L_{i-1})$.
We can prove by contradiction that such a locally optimal scheme can result in an overall optimality(proof in Appendix ~\ref{subsec:local implies overall}):

\begin{algorithm}[htb]
\caption{A non-constructive optimal algorithm scheme for general graphs}
\begin{algorithmic}[1]\label{algo:non-const}
\label{alg:greedy construction}
\STATE Given $L_1, \dotsb, L_{n-1}$, node $n$ constructs $L_n$ that aims at minimizing its own failure probability $\mathcal{E}_n$.

\STATE Given $L_1, \dotsb, L_{n-2}$, and also the fact that node $n$ is greedy,
		 node $n - 1$ constructs $L_{n-1}$ such that the overall loss of him and node $n$ is minimized.

\STATE This process continues. Each $L_i$ greedily minimizes the expected number of wrong nodes after among $\{i, \dotsb, n\}$
		 given $L_1, \dotsb, L_{i-1}$.

\STATE Node 1 knows that all later nodes are ``greedy''.
	Their algorithms $L_2, \dotsb, L_n$ can all be written as a function of $L_1$.
	It then constructs $L_1$ such that $\mathcal{E}$ is minimized.

\STATE Knowing what $L_1$ is, we can backtrack $L_2$, and recursively all the output algorithms $L_i$.
\end{algorithmic}
\end{algorithm}

\begin{thm}
\label{local implies overall}
$L_1, L_2, \dots, L_n$ constructed as in Algorithm \ref{alg:greedy construction} minimizes the expected number of wrong nodes,
	i.e. $\mathcal{E}(L_1, L_2, \dots, L_n)$.
\end{thm}


\subsection {Optimal algorithms for complete graphs}\label{subsec:optimal algorithm}
In this section, we specify the optimal decision-making algorithms for complete graphs and thus provide a general lower bound for our model(by Theorem ~\ref{thm:general lower bound}).
Several intrinsic properties regarding information cascades in complete graph will also be presented.

We start with a lemma showing that optimal algorithm will either reveal \textbf{valid} information or perform Majority Algorithm on all previous guesses.
\begin{lem}
\label{lem:reveal or majority}
In the optimal algorithm, a node either reveals \textbf{valid} information or apply Majority Algorithm on all previous outputs, i.e.
\begin{equation*}
L_i
= \left\{
	\begin{array}{c l}
		s_i & c^{i-1} \in RS_i \\
 		\Maj(c^{i-1}) & c^{i-1} \not \in RS_i
 	\end{array}
\right..
\end{equation*}
\end{lem}
It is worth pointing out several non-trivial points of Lemma \ref{lem:reveal or majority} (See proof in Appendix \ref{subsec:reveal or majority}):
	a) the Majority Algorithm is performed on previous guesses only and ignores its own private signal;
	b) the Majority Algorithm is performed on \emph{all} previous guesses, not only on the valid guesses.
An established result in the proof of Lemma \ref{lem:reveal or majority} is that the Majority Algorithm will cascade on complete graphs, i.e. if a node performs Majority Algorithm, all later nodes will also perform Majority Algorithm.
This implies the existence of a switching point,  where all nodes prior to this point reveal their private signals,
	and all later nodes perform Majority Algorithm based on former nodes' signals.
If we can estimate the position of this switching point, then an estimation of $\mathcal{E}$ can be achieved.

Lemma $\ref{lem:reveal or majority}$ specifies a node's action outside the reveal set.
However, to get an explicit representation of $L_i$, an understanding of the reveal set itself is required.
We introduce the following lemma that fills this gap.

Denote
$
	\diff(c^{i-1})
	=
	(\text{\# 1 in }\text{Valid}(c^{i-1})) - (\text{\# 0 in }\text{Valid}(c^{i-1})),
$
which serves as a criteria to measure the strength of valid information in previous decision vector $c^{i-1} = (c_1, \dotsb, c_{i-1})$.

\begin{lem}
\label{lem:reveal set threshold}
The reveal set of a node $i$ can be explicitly expressed with respect to some parameters $\delta_n(\cdot)$, where
	$RS_i = \{ c^{i-1} :|\diff(c^{i-1})| \ge \delta_n(i) \}$.
\end{lem}
\begin{proof}
This lemma follows from the fact that a node outputs based on the Bayesian probability for the ground truth bit $b$,
	which depends solely upon $\diff(c^{i-1})$.
Given $c^{i-1}$, the Bayesian probability for $b$ is
\begin{equation}
\label{equ:bayesian}
\left\{
\begin{array} {l}
	\Pr[b = 0|c^{i-1}]
	=\frac{1}{1 + ({p\over 1-p})^{\diff(c^{i-1})}}
\\
	\Pr[b = 1|c^{i-1}]
	= \frac{1}{1 + ({1 - p\over p})^{\diff(c^{i-1})}}
\end{array}
\right..
\end{equation}
Lemma $\ref{lem:reveal or majority}$ implies that for each node,
	a) if the previous decision vector convinces it that $b$ equals $\Maj(c^{i-1})$ with high probability,
		it follows the majority of former output decisions;
	b) otherwise, it tries to provide more information by revealing its own private signal.
As implied by Equation $\eqref{equ:bayesian}$, the larger $|\diff(c^{i-1})|$ is, the more likely $b = \Maj(\text{Valid}(c^{i-1}))$.
This lead us to conclude the existence of a threshold $\delta_n(i)$ such that $L_i$ applies Majority Algorithm if and only if $|\diff(c^{i-1})| \ge \delta_n(i)$.
\end{proof}

Finally, given $i$ and $n$ as input, we show how to efficiently derive $\delta_n(i)$ in average $O(\log n)$ time.
Denote $\mathcal{E}(i,d)$ to be the expected number of wrong nodes given that $|\diff(c^{i-1})| = d$.
The idea is to use recursion to derive $\mathcal{E}(i,d)$ for all $i$ and $d$, in the process of which $\{\delta_n(i)|i\}$ may be calculated.
If a node $k$ chooses to reveal its private signal, $\mathcal{E}(k, d)$ is updated as
\begin{equation*}
	\mathcal{E}(k,d)
=
 	q_1\cdot \mathcal{E}(k+1, d+1) + (1-q_1)\cdot \mathcal{E}(k+1, d-1),
\end{equation*}
where $q_1$ is the probability that node $k$'s private signal matches the majority of former guesses.
Similarly, if node $k$ chooses to do Majority Algorithm, $\mathcal{E}(k, d)$ is updated as
\begin{equation*}
	\mathcal{E}(k,d)
=
	q_2\cdot \left(n - {k+d\over 2}\right) + (1-q_2)\cdot {k+d\over 2},
\end{equation*}
where $q_2$ is the probability that the majority of former outputs is correct.
Therefore, we can calculate $\{\mathcal{E}(i,d)\}$ in time $O(n^2)$.
A further improvement can be made by exploiting the properties of $\delta_n()$:
$
\delta_n(i+1)-1\le \delta_n(i)\le \delta_n(i+1)+1,
$
Thus to calculate $\{\delta_n(i)\}$, it suffices to calculate $\{\mathcal{E}(i,d)~|~d < \delta_n(i), i \ge n\}$,
	which requires only $n\cdot \max_{i}\{\delta_n(i)\} =O( n\log n)$ time complexity.
Please see Appendix \ref{subsec:algorithm} for detailed derivation and algorithms.

\begin{lem}
Given $n$ as input, we can calculate the set $\{\delta_n(i)\}$ in $O(n\log n)$ time.
\end{lem}

\subsection{General lower bound for our model}
Finally we analyze the expected number of wrong nodes for the optimal algorithms in complete graphs,
	and provide a general $\Theta(\log n)$ lower bound for the model.
	
\begin{thm}
\label{thm:general lower bound}
The expected number of wrong nodes $\mathcal{E}$ for any topology and any algorithm is at least $\Theta(\log n)$.
\end{thm}

\begin{proof}
It suffices to prove that
	the $\mathcal{E}$ of the optimal algorithms in complete graph is bounded by $\Theta(\log n)$.
In the proof of Lemma \ref{lem:reveal or majority}, we develop the concept of a ``switching point'',
	where all nodes prior to this point reveal valid information and all nodes afterwards perform Majority Algorithm.
Denote $m$ as a random variable of the switching point's position.
We prove that at least one of the following happens:
a) $m\geq (\log_{p/(1-p)}n) / 2$;  or
b) $\mathcal{E}$ is greater than $\sqrt{n}/2$.

If $m < (\log_{p/(1-p)}n) / 2$, then from Equation $\eqref{equ:bayesian}$, we know that the majority of revealed signals are wrong
	with probability at least
	$
	({1 + (p/(1-p))^m})^{-1}
	>
	{n^{-0.5}}/{2}.
	$
So the expected number of wrong nodes is lower bounded by $ {(n-m)} / {\sqrt{n}}$ which is asymptotically greater than ${\sqrt{n}}/2$.
Therefore, for the optimal output algorithms in complete graph, $\mathcal{E}$ is at least
\begin{equation*}
	\min \Big((1 - p) \log_{p/(1-p)}n \Big/ 2,  \sqrt{n} \Big/2 \Big)
	=
	\Omega(\log n).
\end{equation*}
\end{proof}

\section {Optimal K-Layer Topology}\label{sec:k layer}

In section 3, we show that the optimal $\mathcal{E}$ for random graphs is $\Theta(\log n)$, which is asymptotically the same as the general lower bound in Theorem $\ref{thm:general lower bound}$.
Yet the question remains what is the actual optimal topology for the Majority Algorithm.

In this section, we propose a family of layer graphs and search for the optimal topology among this family.
We claim, without proof, that the optimal topology of layer graphs is actually the optimal topology for Majority Algorithm.

\begin{figure}
\centering
\includegraphics[scale=0.3]{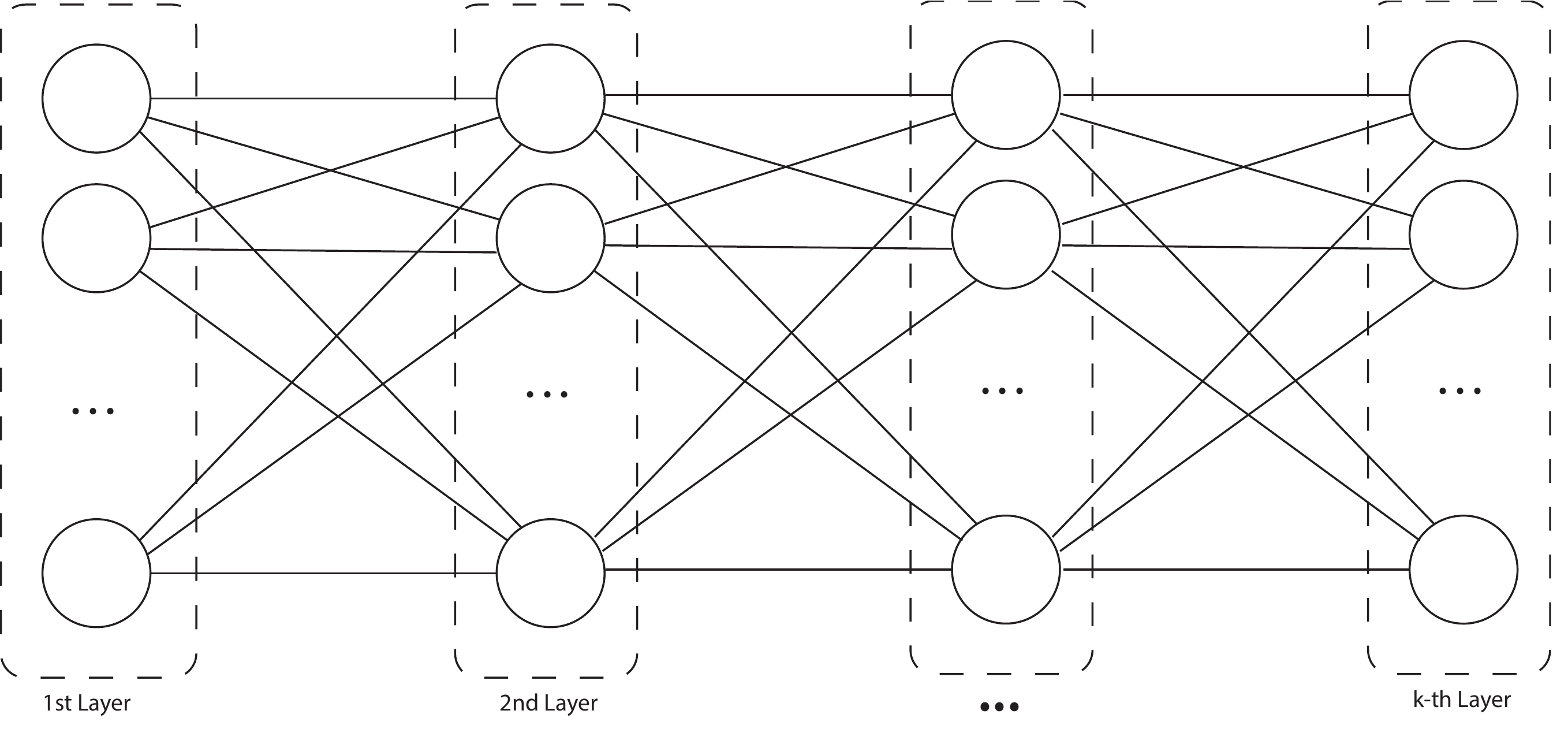}
\caption{A $k$-layer graph}
\label{fig:K-Layer}
\end{figure}

To find the optimal topology, it helps to first understand the hidden insights behind small overall $\mathcal{E}$.
In the optimal algorithms for complete graphs, nodes first judge the strength of the current cascade,
	and then decide whether to follow the cascade or reveal their own signals to strengthen the cascade.
Such a \emph{think-before-acting} way of decision-making guarantees the correctness probability of any established cascade, and thus results in good overall performance.
We hope to know whether such \emph{think-before-acting} could make it possible for Majority Algorithm to achieve optimality simply by adjusting the topology.
This inspires us the following definition of layer graphs.

\begin{defn} {\bf (Definition of layer graphs)}
\\
A graph is said to have $k$ layers if it can be separated into $k$ disjoint groups, $S_1, \dotsb, S_k$, where any node in group $S_i$ is connected to and only to all nodes in $S_{i-1}$. See Figure $\ref{fig:K-Layer}$ for an example.
\end{defn}

\begin{rmk}
\label{rmk2}
Given a $k$-layer graph $G$, we consider how nodes perform in $G$.
First of all, similar to the optimal algorithms in complete graph,
	we have $|S_1|$ many nodes revealing \textbf{valid} information at the very front.
If there exists a cascade in $S_1$ (the number of one choices outmatches another by at least two),
	then all later nodes follow this cascade.
Otherwise, nodes in $S_2$ reveal their private signals.
This process continues until a cascade happens in some layer.
In this sense, layer graphs do contains the think-before-acting way of decision-making.
\end{rmk}

In the following sections, we find the optimal topology among layer graphs,
	and show that the expected number of wrong nodes $\mathcal{E}$ for such optimal topology is also $\Theta(\log n)$.
This optimal $\mathcal{E}$ will be compared to previous bound and results,
	from which we will be able to glimpse the limit of Majority Algorithm.
Throughout this section, if not otherwise mentioned, we will assume that the output algorithm is Majority Algorithm.
We denote the expected number of wrong nodes on a $k$-layer topology $(S_1, \dotsb, S_n)$ as $\mathcal{E}(|S_1|, \dotsb, |S_n|)$.

\subsection {Optimal topology for layer graphs}

Remark $\ref{rmk2}$ provides an intuitive way to calculate the $\mathcal{E}$ of any layer graph.
Given $i$ independent signals, we denote $p_w(i)$ as the probability that these signals generate a wrong cascade,
	and $p_n(i)$ as the probability that these signals does not generate cascade in either side.
A recursion regarding $\mathcal{E}$ of any layer graph can be shown to be:
\begin{equation}
\label{equ:recursive}
	\mathcal{E}(a_1, \dotsb, a_k)
=
	(1-p)\cdot a_1 + p_w(a_1)\cdot (n - a_1) + p_n(a_1)\cdot \mathcal{E}(a_2, \dotsb, a_k),
\end{equation}
where $(1-p)a_1$ is the expected number of wrong nodes among the first layer,
	$n-a_1$ is the expected number of wrong nodes among later layers under wrong cascade,
	and $\mathcal{E}(a_2, \dotsb, a_k)$ is the expected number of wrong nodes of later layers under no cascade.
By extending the recursive term, we can simplify Equation $\eqref{equ:recursive}$ into
\begin{equation}
\label{equ:loss of k-layer}
	\mathcal{E}(a_1, \dotsb, a_k)
=
	\sum_{i=1}^{k} \Big(
		\prod_{j=1}^{i-1}p_n(a_j) \cdot \Big(  (1-p)\cdot a_i + p_w(a_i)\cdot (n - \sum_{j=1}^{i-1}a_j) \Big)
	\Big).
\end{equation}
\begin{rmk}
\label{rmk:estimation}
Our goal is to estimate $\mathcal{E}$ of the optimal layer graph
	to the $\Theta(\log n)$ level.
Any approximation of $\mathcal{E} + o(\log n)$ would be satisfying.
This relaxation releases us from getting an exact optimum and allows us to make proper adjustments that greatly reduce the difficulty of the calculation.
For example, in the derivation of Equation $\eqref{equ:recursive}$ and $\eqref{equ:loss of k-layer}$, we assume without loss of generality that each layer has an even number of nodes.
This will result in $O(1)$ changes in the optimal parameters, which is tolerable.
\end{rmk}

To find a set of parameters $\{k, a_1, \dotsb, a_k\}$ such that $\mathcal{E}(a_1, \dotsb, a_k)$ is minimized, we need the following basic steps:
\begin{itemize}
\item We first show that in Equation $\eqref{equ:recursive}$,
		the contribution of $p_n(a_1)\mathcal{E}(a_2, \dotsb, a_k)$ is limited and
		may be discarded without much change to the optimal parameters.
	Therefore, it suffices to consider the optimization of $a_1$ in the equation
	\begin{equation}
	\label{equ:equ1}
		\arg\min_{a_1} f(a_1) = \arg\min_{a_1} \Big( (1-p)\cdot a_1 + p_w(a_1)\cdot (n - a_1) \Big).
	\end{equation}
\item We then solve the equation $f(x+1) - f(x) = 0$, which has a unique solution.
	It can be shown that $f(a_1)$ first decreases then increases with respect to $a_1$.
	Thus the solution for $f(x+1) - f(x) = 0$ offers an approximation to the optimal $a_1$ with only $O(1)$ error.
\item After solving the optimal size of the first layer, we can apply this method recursively to calculate the optimal size of all layers.
\end{itemize}

The proof for the above three results are complex and brute-force.
We list the formal lemmas and theorems here and leave the detailed proof to the appendix.
First, we present a lemma that addresses result 2.
It is worth pointing out that Equation $\ref{equ:equ1}$ is the expected loss when we have only two layers,
	with the first layer of size $a_1$ and the second layer of size $n - a_1$.
Therefore, result 2 is equivalent to finding an optimal topology among two layer graphs.
For convenience, we denote $s = 1 / (4p(1-p))$.

\begin{lem}
\label{lem:twolayer}
For the optimal topology among two layer graphs, its first layer has size
\begin{equation*}
\log_{s}n-\log_s( \log_s n)/2 + O(1).
\end{equation*}
\end{lem}
\begin{proof}
See Appendix $\ref{subsec:two-layer optimal}$.
\end{proof}

\begin{thm}
\label{thm:optimal layer topology}
The optimal layer topology has $k = n/\log_s n + o(n/\log_s n)$ many layers.
The first layer has $a_1 = \log_{s}n$ many nodes.
The size of layer $i$ may be written as a recursion of the size of layer $i-1$,
\begin{equation}
\label{equ:optimal parameters}
	a_i
	\sim
	\log_s(s^{a_{i-1}}-a_{i-1}).
\end{equation}
In other words, the optimal topology satisfies the following structural properties:
\begin{itemize}
\item The sizes of layers gradually decrease, and the number of layers with size $\log_s n - i$ is ${(s-1)n / (s^{i+1}(\log_{s}n - i))}$.
\item The first ${(s-1)n / (s\log_{s}n)}$ layers have size $\log_{s} n$.
\item The following ${(s-1)n / (s^{2}(\log_{s}n - 1))}$ layers have size $\log_s n - 1$, and so on.
\end{itemize}
\end{thm}

\begin{proof}
See Appendix $\ref{subsec:optimality all layers}$.
\end{proof}

We believe that the layer topology provided in Theorem $\ref{thm:optimal layer topology}$ is
	actually the optimal topology for Majority Algorithm.
However, we have not yet come up with any rigorous proof to verify our conjecture.
We will leave this as an open problem for future work.
\begin{cjt}
\label{cjt:open problems}
The layer topology provided in Theorem $\ref{thm:optimal layer topology}$ is
	the optimal topology for Majority Algorithm.
\end{cjt}

\subsection {Experiments on the optimal parameters}

For layer graphs, Equation $\eqref{equ:optimal parameters}$ can be used to calculate the optimal parameters with high precision.
The analysis in Appendix $\ref{subsec:two-layer optimal}$ implies that the real optimal $k^{opt}$ must be constant away from our calculated $k^{opt'}$, i.e.
\begin{equation*}
	k^{opt} - k^{opt'}
=
	O(1).
\end{equation*}
Thus the results in Theorem $\ref{thm:optimal layer topology}$ is a very tight approximation, which works flawlessly
	if we only seek to analyze the complexity of $\mathcal{E}$.
However, in real life, users might wish to achieve the exact optimal parameters.
In this case, the constant error in our equations can not be neglected.
Here, we introduce an algorithm that searches for the exact optimal topology in average $O(1)$ run time.

Recall that layer graphs satisfy the following properties.
For a $k$-layer structure $(a_1, \dotsb, a_k)$,
	if the first layer cascades, the rest of the nodes follow this cascaded result.
Otherwise, the rest of the nodes become equivalent to a $(k-1)$-layer structure, with each layer's size being $(a_2,\dotsb,a_k)$ (or $(a_2+1,\dotsb,a_k)$ if $a_1$ is odd).
Therefore, for a fixed $a_1$, the optimal $k$ and $a_2, \dotsb, a_k$ should be chosen such that,
\begin{itemize}
\item If $a_1$ is even,
\begin{equation}
\label{equ:recur1}
	(k-1, a_2, \dotsb, a_k)
=
	\arg\max_{k', a'_1, \dotsb, a'_k}
		\mathcal{E}\Big(
			a'_1, \dotsb, a_k' \Big| \sum_{i=1}^{k'} a'_i = n - a_1
		\Big).
\end{equation}
\item If $a_1$ is odd,
\begin{equation}
\label{equ:recur2}
	(k-1, a_2 + 1, \dotsb, a_k)
=
	\arg\max_{k', a'_1, \dotsb, a'_k}
		\mathcal{E}\Big(
			a'_1, \dotsb, a_k' \Big| \sum_{i=1}^{k'} a'_i = n - a_1+ 1
		\Big).
\end{equation}
\end{itemize}
This implies that given the optimal layer topologies for all $n' < n$, calculating the optimal layer topology for $n$ should be easy.
We can simplify the problem into an optimization over $a_1$, instead of the optimization over many parameters.
A further analysis shows that the optimal layer topology for $n$ nodes and $n+1$ nodes cannot differ by too much.
More specifically, denote the optimizing parameter for $n$ nodes as $(k_n, a^n_1, \dotsb, a^n_{k_n})$, then $|a^n_1 - a^{n+1}_1| \le 1$.
Using this, together with the recursive idea in Equation $\ref{equ:recur1}$ and $\ref{equ:recur2}$, 	
	we can design an algorithm that runs in time $O(n)$, and find the optimal layer topology for all $n' \le n$.
The amortized running time of this algorithm is only $O(1)$.

\section{Conclusion}

In this paper, we discussed information cascades on various network topologies.
We provide a non-constructive optimal algorithm scheme for general graphs, solve the scheme for complete graph and achieve a general lower bound for our model.
We also studied Majority Algorithm in random graphs and layer graphs, the minimal $\mathcal{E}$ of which was shown to be asymptotically the same with our general lower bound.
From the experiment results, a gap between the general lower bound and layer graphs can be observed.
 We believe this to be a result of the difference in the model setting, i.e. Majority Algorithm is weaker than optimal general algorithms.

Future work in this area may include the study of the following scenarios.
\begin{itemize}
\item The nodes' order of decision-making is no longer fixed and given, but instead randomly sampled from all permutations.
\item The topology is fixed and we are only able to add or remove a fixed portion of the edges.
		The goal is to minimize $\mathcal{E}$ under this constraint.
\item -	\emph{Plant nodes} in the network. These nodes could sacrifice themselves to reveal their true private signal. How should a topology designer control and position these plant nodes in the topology?
\end{itemize}


%
%
%
%
%
%
%
\bibliography{infocas}{}
\bibliographystyle{plain}

\clearpage
\appendix
\section{Proof on Random Network}

\subsection{Proof of Lemma \ref{lem:first logn/q deciding}}
\label{subsec:first logn/q deciding}
If among the first $\log n/q$ nodes,
	only a small constant fraction $f < 0.5$ output wrongly,
	then for the later $n - (\log n/q)$ nodes, the expected number of wrong outputs is at most $O(1)$.

\begin{proof}
We show that the probability of having any incorrect nodes among the later $n - (\log n/q)$ nodes is at most $O(1/n)$.
Therefore, the expected number of wrong outputs is at most $n\cdot O(1/n) = O(1)$.
Let $X_i$ $(i > \log n/q)$ denote the event that node $i$ outputs correctly.
Using Lemma \ref{lem:random exponential bound}, we can bound the conditional probability
\begin{equation*}
	\Pr\Big[
		X_{i+1}~\Big|~ X_{\log n/q}, \dotsb, X_{i}
	\Big]
\ge
	1 - e^{-iq(\sqrt{f} - \sqrt{1-f})^2}
=
	1 - n^{-\Theta(1)}.
\end{equation*}
Here, the exponential above $n$ is determined by $f$, $p$ and the base of the $\log()$ function.
By making the base of the $\log()$ function small enough, the bound can be set to $1 - O(n^{-c})$ for any positive constant $c$.
Let $c = 3$, the expected number of wrong outputs for the later $(n - \log n/q)$ nodes is upper bounded by
\begin{align*}
&
	0\cdot \Pr[X_{\log n/q}, \dotsb, X_n] + (n - \log n/q)(1 - \Pr[X_{\log n/q}, \dotsb, X_n]) \\
=~ &
	(n - \log n/q)(1 - \Pi_{i = n - \log n/q}^{n-1}\Pr[X_{i+1} | X_{\log n/q}, \dotsb, X_i]) \\
\le~ &
	(n - \log n/q)(1 - (1 - n^{-2})^n)
	\le
	O(1).
\end{align*}
\end{proof}

\subsection{Proof of Lemma \ref{lem:random exponential bound}}
\label{subsec:random exponential bound}

If a constant fraction $f > 0.5$ of the first $i$ nodes are correct,
	node $i+1$'s failure probability $\mathcal{E}$ can be bounded within $e^{-\Theta(iq)}$.

\begin{proof}
Recall that $\text{Bin}(n,q)$ denotes a binomial distribution with parameters $n$ and $q$,
	and that $\text{Bin}(fi, q)$ (resp. $\text{Bin}((1-f)i, q)$) can be used to denote the number of correct (resp. wrong) neighbors for node $i+1$.
By definition of Majority Algorithm, we have
\begin{equation*}
	1 - \mathcal{E}_{i+1}
\ge
	\Pr\Big[
		\text{Bin}(fi, q) \ge \text{Bin}((1-f)i, q) + 2
	\Big]
\sim
	\Pr \Big[
		\text{Bin}(fi, q) \le \text{Bin}((1-f)i, q)
	\Big ].
\end{equation*}
We create a symbol $X_j$ for each node $j \in [1, i]$, where
\begin{equation*}
\text{if node $j$'s decision is correct, }~~
	X_j =
	\left \{
	\begin{array} {c l}
		-1 & \text{with probability $q$.} \\
		0 & \text{with probability $1-q$.}
	\end{array}
	\right.
\end{equation*}
\begin{equation*}
\text{if node $j$'s decision is incorrect, }~~
	X_j =
	\left \{
	\begin{array} {c l}
		1 & \text{with probability $q$.} \\
		0 & \text{with probability $1-q$.}
	\end{array}
	\right.
\end{equation*}
By this, we transform the original problem into a deterministic form.
\begin{equation*}
	\Pr \Big[
		\text{Bin}(fi, q) \le \text{Bin}((1-f)i, q)
	\Big]
=
	\Pr\Big[
		\sum_{j=1}^i X_j \ge 0
	\Big].
\end{equation*}
Denote $X = \sum_{j=1}^i X_j$, then
\begin{align*}
	\Pr[X \ge 0]
& =
	\Pr[e^{X} \ge 1 ]
 	\le
	\min_{t>0} \E(e^{tX})
&
	(\text{by Markov Inequality}) \\
&=
	\min_{t>0}\prod_{j = 1}^{i}\E(e^{tX_j})
&
	(\text{$X_i$ are independent to each others}) \\
& =
	\min_{t > 0} (qe^{-t} + (1-q))^{fi}\cdot (qe^{t} + (1-q))^{(1-f)i}.
\end{align*}
If we apply the logarithmic function to both sides, then
\begin{equation}
\label{equ:exponential}
	\ln (\Pr[X \le 0])
\le
	i\min_{t>0} \Big( f \ln(qe^{-t} + (1-q)) + (1-f)\ln (qe^{t} + (1-q)) \Big).
\end{equation}
Using $\ln(1+x) < x$, we can relax Equation \eqref{equ:exponential} into simpler form.
\begin{equation}
\label{equ:takeback}
	\ln (\Pr[X \le 0])
\le
	i\min_{t>0} \Big( qf(e^{-t} - 1) + (1-f)q(e^{t} - 1) \Big).
\end{equation}
Let $t = \ln (f/(1-f)) / 2 > 0$, and take this value back to Equation \eqref{equ:takeback}.
We get
\begin{equation*}
	\ln (\Pr[X \le 0])
\le
	qi (2\sqrt{f(1-f)} - f - (1-f))
=
	-qi(\sqrt{f} - \sqrt{1-f})^2.
\end{equation*}
Thus the probability of node $i+1$'s failure is upper bounded by $e^{-qi(\sqrt{f} - \sqrt{1-f})^2}.$
\end{proof}

\subsection{Proof of \ref{lem:random general bound}}
\label{subsec:random general bound}

If a constant fraction $f > 0.5$ of the first $i$ nodes are correct, s.t.
$
	{f\over 1-f} \ge \sqrt{q\over 1-q},
$
then node $i+1$'s failure probability is upper bounded by $p$.

\begin{proof}
Denote random variable $X = \text{Bin}(fi, q) - \text{Bin}((1-f)i, q)$.
By definition of Majority Algorithm, node $i+1$ outputs a correct output with probability
\begin{equation*}
	1 - \mathcal{E}_{i+1}
=
	\Pr[ X \ge 2]
	+
	p\cdot \Pr[ 1 \ge X \ge -1].
\end{equation*}
If we can show that
$
	(1-p)\cdot\Pr[ X \ge 2 ]
	\ge
	p\cdot\Pr[ X \le -2 ],
$
then it follows that,
\begin{align*}
	{\Pr[\text{correct}]\over 1 - \Pr[\text{correct}]}
=
	{\Pr[ X \ge 2]+ p\cdot \Pr[ 1 \ge X \ge -1] \over \Pr[ X \le -2]+ (1-p)\cdot \Pr[ 1 \ge X \ge -1]}
\ge
	{p\over 1-p}
\implies
	\Pr[\text{correct}]
\ge
	p.
\end{align*}
To prove such an inequality, we extend the probability equations into a sum of specific cases,
\begin{align*}
	{\Pr \Big[ X \ge 2 \Big] \over \Pr\Big[ X \le -2 \Big]}
& =
	{\sum_{l_1 - l_2 \ge 2} \Big( \Pr\Big[ \text{Bin}(fi, q) = l_1 \Big]\cdot \Pr\Big[ \text{Bin}((1-f)i, q) = l_2 \Big] \Big)
	\over
	\sum_{l_1 - l_2 \ge 2} \Big( \Pr\Big[ \text{Bin}(fi, q) = l_2 \Big]\cdot \Pr\Big[ \text{Bin}((1-f)i, q) = l_1 \Big] \Big) }
\\
& =
	{\sum_{l_1 - l_2 \ge 2} \Big( {fi\choose l_1} q^{l_1}(1-q)^{fi - l_1}\cdot {(1-f)i\choose l_2} q^{l_2}(1-q)^{(1-f)i - l_2} \Big)
	\over
	\sum_{l_1 - l_2 \ge 2} \Big( {fi\choose l_2} q^{l_2}(1-q)^{fi - l_2}\cdot {(1-f)i\choose l_2} q^{l_1}(1-q)^{(1-f)i - l_1} \Big) }
\\
& =
	{\sum_{l_1 - l_2 \ge 2} {fi\choose l_1} \cdot {(1-f)i\choose l_2}
	\over
	\sum_{l_1 - l_2 \ge 2} {fi\choose l_2} \cdot {(1-f)i\choose l_1}  }
\\
& \ge
	\min_{l_1 - l_2\ge 2} \Big( {fi\choose l_1} \cdot {(1-f)i\choose l_2}\cdot {fi\choose l_2}^{-1} \cdot {(1-f)i\choose l_1}^{-1} \Big)
\\
& =
	\min_{l_1 - l_2\ge 2} {fi - l_2 \choose l_1 - l_2} \cdot {(1-f)i - l_2 \choose l_1 - l_2}^{-1} \\
& \ge
	({f\over 1-f})^{l_1 - l_2}
	\ge
	({f\over 1-f})^2
	\ge {p\over 1-p}.
\end{align*}
This completes our proof.
\end{proof}

\subsection {Proof of Lemma \ref{lem:majority among first logn/q}}
\label {subsec:majority among first logn/q}

Let $\delta = {1\over 2}(p + {\sqrt{p}\over \sqrt{p} + \sqrt{1-p}})$.
There exists connection probability $q_{opt} = \Theta(1/\log n)$ such that the first $\Theta(\log n/q)$ nodes
	contains at least $\delta$ portion of correct outputs with probability $1 - O(n^{-1}\log n)$.

\begin{proof}
Divide the first $\Theta(\log n/q)$ nodes into $\Theta(\log n)$ segments, where each segment contains $\Theta(1/q)$ many nodes.
We analyze each segment's number of wrong outputs independently and apply Union Bound to get a final bound.
The proof can be divided into two parts.
\begin{itemize}
\item We show that there exists $q_1 = \Theta(1/\log n)$ and some constant $a > 0$, such that for any $q < q_1$,
	the first $a/q$ nodes contain $\delta$ portion of correct outputs with probability at least $1 - O(n^{-1})$.
\item Given $a$ and $q_1$, we divide the first $\Theta(\log n / q)$ nodes into various segments.
	The first segment contains $a/q$ nodes while all later segments contain $b/q$ nodes, ($b$ will be determined later).
	Denote $X_i$ as the event that the $i^{th}$ segment contains $\delta$ portion of correct outputs.
	We prove that there exists $q_2 = \Theta(1/\log n) < q_1$ and some constant $b > 0$ such that
	\begin{equation*}
		\forall i, \Pr[X_{i}~|~X_1, \dotsb, X_{i-1}] = 1 - O(n^{-1}).
	\end{equation*}
\end{itemize}
These two claims form an induction analysis.
The first claim implies that the first segment are majority-correct with high probability.
And the second claim implies that if the first $i$ segments are majority-correct,
	then segment $i+1$ will also be majority-correct with high probability.
Together, they lower bound the probability that the first $\Theta(\log n / q)$ nodes contains $\delta$ portion of correct outputs by
\begin{equation*}
	\Pr[ \forall i, X_i]
=
	\prod_{i} \Pr[X_{i}~|~X_1, \dotsb, X_{i-1}]
=
	(1 - O(n^{-1}))^{\Theta(\log n)} = 1 - O(n^{-1}\log n),
\end{equation*}
which is exactly what we desire.
It suffices to prove the two claims.

Let us first address the performance of the first segment, i.e. prove the first claim.
Any node in the first segment expect to see at most $(a/q)*q = a$ neighbors.
When $a$ is small enough, the node will be separated with high probability.
Mathematically, the probability of a node having zero or one neighbor is bounded by
$
	(1-q)^{a/q} + {a/q \choose 1}q(1-q)^{a/q - 1}
\sim
	e^{-a}.
$
Therefore, any node in the first segment has at least $pe^{-a}$ probability of outputting correctly.
Since $\delta < p$, we can set $a$ such that
$
	a < \ln q - \ln \delta
	\implies
	pe^{-a} > \delta.
$
For simplicity, we denote $pe^{-a}(1-\epsilon) = \delta$, where $\epsilon$ is some positive constants.
Notice that in the above analysis, the correlations between nodes are ignored.
Any node in the first segment will have probability $pe^{-a}$ of being correct despite the performance of other nodes.
Therefore, we can use Chernoff Bound to show that the first segment
	contains $\delta$ portion of correct outputs with probability at least $e^{-\epsilon^2a/(3q)}$.
Set $q < q_1 = \epsilon^2a / (3\ln n)$, then a bound of $1 - O(n^{-1})$ can be achieved.
This completes the proof of the first claim.

We now turn to the second claim.
Suppose all of the first $i$ segments have at least $\delta$ portion of correct outputs,
	then for any node in segment $i+1$, there are at least
\begin{equation*}
	{\#(\text{correct outputs in the first $i$ segments}) \over \#(\text{nodes in the first $i+1$ segments})}
\ge
	{\delta(a + b(i-1)) \over a + bi }
\ge
	{\delta a\over a + b}
\end{equation*}
fraction of correct outputs in front.
Since $\delta > \sqrt{p}/(\sqrt{p}+\sqrt{1-p})$, there exists constant $b > 0$, such that
$
	{\delta a\over a + b}
\ge
	{\sqrt{p}\over \sqrt{p}+\sqrt{1-p}}.
$
By Lemma $\ref{lem:random general bound}$, any node in segment $i+1$ has at least probability $p$ of outputting a correct output.
Again, this analysis is independent of the correlation between nodes.
Thus Chernoff Bound can be applied to bound the probability.
This completes the proof of the second claim.
\end{proof}

\subsection{Proof of Lemma \ref{lem:loss of former nodes}}
\label{subsec:loss of former nodes}
If $q = q_{opt}$, the first $\Theta(\log n/q)$ nodes' expected number of wrong outputs is $\Theta(\log n)$.

\begin{proof}
By linearity of expectation, the expected number of wrong outputs among the first $\Theta(\log n/q)$ nodes
	is the sum of the failure probability for each node.
Denote $Y$ as the event that all segments have at least $\delta$ portion of correct outputs.
If $Y$ happens, then as argued in Appendix \ref{subsec:majority among first logn/q},
	each node will have at least $\delta' = {\sqrt{p}/(\sqrt{p}+\sqrt{1-p})}$ portion of correct nodes in front.
By Lemma \ref{lem:random exponential bound}, the loss of the first $\Theta(\log n/q)$ nodes conditioned on $Y$ is upper bounded by
\begin{align*}
	\sum_{i=1}^{\Theta(\log n/q)} (\mathcal{E}_i|Y)
& =
	\sum_{i=1}^{a/q} (\mathcal{E}_i|Y) + \sum_{i=a/q+1}^{\Theta(\log n/q)} (\mathcal{E}_i|Y)
	\le
	{a \over q} + \sum_{i=1}^{\Theta(\log n/q)} e^{-qi(\sqrt{\delta'} - \sqrt{1 - \delta'})^2} \\
& \le
	{a \over q} + { 1 \over 1 - e^{-q(\sqrt{\delta'} - \sqrt{1 - \delta'})^2} }
	\le
	{a \over q} + {1 \over q(\sqrt{\delta'} - \sqrt{1 - \delta'})^2 } \\
& =
	\Theta ({1\over q})
	=
	\Theta(\log n).
\end{align*}
Therefore, the expected number of wrong outputs among the first $\Theta(\log n/q)$ nodes is
\begin{align*}
	\sum_{i=1}^{\Theta(\log n/q)} \mathcal{E}_i
& =
	\Pr[Y]\cdot \sum_{i=1}^{\Theta(\log n/q)} (\mathcal{E}_i|Y) + \Pr[\overline Y]\cdot \sum_{i=1}^{\Theta(\log n/q)} (\mathcal{E}_i|\overline Y) \\
& \le
	\sum_{i=1}^{\Theta(\log n/q)} (\mathcal{E}_i|Y) +  \Theta({\log n \over n})\cdot \Theta({\log n \over q})
	\le
	\Theta(\log n).
\end{align*}
This completes our proof.
\end{proof}

\section{Proof for General Lower Bound}

\subsection{Proof of Lemma \ref{lem:selfMaj}}
\label{subsec:selfMaj}

For a node to minimize its failure probability,
	the best algorithm is to perform a Majority Algorithm on
	$\{$all previous \textbf{valid} choices, the node's private signal $\}$.

\begin{proof}
Given previous outputs sequence $c^{i-1}$ and a private signal $s_i$,
	suppose vector $(c^{i-1}, s_i)$ contains $n_0$ valid choices of 0 and $n_1$ valid choices of 1.
We calculate and compare the Bayesian probability for the ground truth value $b$.
The node should output 1 if and only if the Bayesian probability of $b=1$ is higher than $b=0$.
It will be shown that
\begin{equation*}
	\Pr[b = 0 | c^{i-1},s_i] \ge \Pr[b = 1|c^{i-1},s_i]
	~\text{ if and only if }~
	n_0 \ge n_1,
\end{equation*}
thus proving the lemma.
To compare the two conditional probability, calculate
\begin{align*}
	{\Pr[b = 0~|~c_1, \dotsb, c_{i - 1}, s_j ]
	\over
	\Pr[b = 1~|~c_1, \dotsb, c_{i - 1}, s_j ]}
= &
	~~{\Pr[c_1, \dotsb, c_{i - 1}, s_j ~|~ b = 0]
	\over
	\Pr[c_1, \dotsb, c_{i - 1}, s_j ~|~ b = 1]}
	~~
	(\text{By Bayesian Theorem})
\\
= &
	~~{\Pr[s_i|b = 0] \over  \Pr[s_i|b = 1]}
	\cdot
		{\prod_{j = 1 \to i-1} \Pr[c_j|c_1, \dotsb, c_{j - 1}, b = 0]
		\over
		\prod_{j = 1 \to i-1} \Pr[c_j|c_1, \dotsb, c_{j - 1}, b = 1]},
\end{align*}
where the last equality holds by conditional expectation.
If $(c_1, \dotsb, c_{j - 1})$ is not in the reveal set of node $j$,
	then by definition, $c_j$ would be independent of node $j$'s signal.
If $(c_1, \dotsb, c_{j - 1})$ is in the reveal set, $c_j$ would be the same as node $j$'s signal.
Therefore,
\begin{equation*}
	\Pr[c_j|c_1, \dotsb, c_{j - 1}, b] =
	\left\{
	\begin{array} {c l}
	1 & \text{ if $c_j$ is not valid } \\
	p\cdot \I[c_j = b] + (1-p)\cdot \I[c_j\neq b] & \text{ if $c_j$ is valid } \\
	\end{array}
	\right..
\end{equation*}
Plot this into the above equation, we get	
\begin{align*}
	{\Pr[b = 0~|~c_1, \dotsb, c_{i - 1}, s_j ]
	\over
	 \Pr[b = 1~|~c_1, \dotsb, c_{i - 1}, s_j ] }
=
	{\prod_{j\text{ is valid}} \Pr[s_j~|~ b = 0]
	\over
	\prod_{j\text{ is valid}} \Pr[s_j~|~ b = 1]}
=
	{ p^{n_0}(1-p)^{n_1} \over p^{n_1}(1-p)^{n_0} }
=
	({p \over  1 - p})^{n_0 - n_1}.
\end{align*}
which is less than 1 if and only if $n_0 < n_1$ and completes the proof.
\end{proof}

\subsection{Proof of Theorem \ref{local implies overall}}
\label{subsec:local implies overall}

$L_1, L_2, \dots, L_n$ constructed as in Section \ref{Construction} minimizes $\mathcal{E}(L_1, L_2, \dots, L_n)$ for complete graphs.

\begin{proof}
Assume the contrary that there exists a set of algorithm $L_1', L_2', \dots, L_n'$ s.t.
\begin{equation*}
	\mathcal{E}(L_1', L_2', \dots, L_n')
	>
	\mathcal{E}(L_1, L_2, \dots, L_n).
\end{equation*}
We denote $k$ to be the largest number in $\{L'_i\}$ where the construction requirement is not met,
\begin{equation*}
	k
\gets
	\sup \Big\{
		k
		~\Big|~
		\exists c^{k-1} \in \{0,1\}^{k-1}, s_k\in \{0,1\}
		\text{ ~s.t.~}
		L_k'(c^{k-1}, s_k) \not= GC(L_1', L_2', \dots, L_{k-1}')(c^{k-1}, s_k)
	\Big\}.
\end{equation*}
Consider a new sequence of algorithms, were
\begin{equation*}
	L_i''
	=
	\left\{
		\begin{array}{l l}
			L_i', & i=1,2,\dots,k-1  \\
			GC(L_1'', L_2'', \dots, L_{i-1}''), & i=k,k+1,\dots,n
		\end{array}
	\right..
\end{equation*}
We then show that $L'$ performs no better than $L''$, i.e.
$
	\mathcal{E}(L_1', L_2', \dots, L_n')
	\leq
	\mathcal{E}(L_1'', L_2'', \dots, L_n'').
$
Minus the two items and consider their difference.
\begin{align*}
	\Delta W
& =
	\mathcal{E}(L_1'', L_2'', \dots, L_n'') - \mathcal{E}(L_1', L_2', \dots, L_n')
\\
& =
	\mathbb E_{s_1, \dots, s_n} \Big[
		\sum_{i=1}^n I_{c_i''=b} - \sum_{i=1}^n I_{c_i'=b}
		~\Big|~
		c_i' \gets L_i'(c'^{i-1}, s_i),
		c_i'' \gets L_i''(c''^{i-1}, s_i), \forall i \in [n]
	\Big].
\end{align*}
As $L_i'=L_i''$ for $i=1,2,\dots,k-1$, we can discard the first $k$ items in the summation and get,
\begin{equation*}
	\Delta \mathcal{E}
=
	\sum_{i=k}^n \Pr_{s_1, s_2, \dots, s_n} [c_i'' = b|c''^{i-1}] -
	\sum_{i=k}^n \Pr_{s_1, s_2, \dots, s_n} [c_i' = b|c'^{i-1}].
\end{equation*}
By the definition, we know that greedy construction maximize $\sum_{i=k}^n \Pr_{s_1, s_2, \dots, s_n} [c_i = b|c^{i-1}]$.
Thus $\Delta SW \geq 0$.
Define the $degree$ of a sequence of algorithms $L_1, L_2, \dots, L_n$ to be
\begin{multline}
	degree(l')
=
	\max \Big\{ k
		~ \Big| ~
			\exists c^{k-1} \in \{0,1\}^{k-1}, s_k \in \{0,1\} \text{ s.t. }
			\\
			L_k'(c^{k-1}, s_k) \not= GC(L_1', L_2', \dots, L_{k-1}')(c^{k-1}, s_k)
		\Big\},
\end{multline}
then for any algorithm sequence $l_1, l_2, \dots, l_n$ with $degree(l)>0$, we can always find another sequence $(l_1', l_2', \dots, l_n')$ such that
	$degree(l') \leq degree(l) - 1$ and $\mathcal{E}(l')\geq \mathcal{E}(l)$.
Therefore, there must exist a sequence $L_1^*, L_2^*, \dots, L_n^*$ such that
\begin{equation*}
\mathcal{E}(L_1^*, L_2^*, \dots, L_n^*) \geq \mathcal{E}(L_1', L_2', \dots, L_n')
~~\&\&~~
degree(L_1^*, L_2^*, \dots, L_n^*) = 0.
\end{equation*}
This implies that $L_i^*$ satisfies all the constraints in the non-constructive scheme, thus completing our proof.
\end{proof}

\subsection {Proof of Lemma \ref{lem:reveal or majority}}
\label{subsec:reveal or majority}
In the optimal algorithm, a node either reveals \textbf{valid} information or apply Majority Algorithm on all previous outputs, i.e.
\begin{equation*}
L_i
= \left\{
	\begin{array}{c l}
		s_i & c^{i-1} \in RS_i \\
 		\Maj(c^{i-1}) & c^{i-1} \not \in RS_i
 	\end{array}
\right..
\end{equation*}

\begin{proof}
We organize the proof into proving several sub-results,
	where each sub-result leads us one  step closer towards the lemma.

(1) Only valid choices can influence later nodes.

We first show that for any $i$, node $i+1, i+2, \dots, n$ does \textbf{not} care about dimensions of $c^{i}$ outside $\text{Valid}(c^{i})$.
The proof is based on induction.
We start with node $n$ and gradually move forward to prove this claim for all $i\in [n]$.
In the greedy scheme, node $n$ uses $L_n$ to minimize its own failure probability $\mathcal{E}_n$.
By Lemma $\ref{lem:selfMaj}$, node $n$ performs a Majority Algorithm on $\text{Valid}(c^{n-1})$ and its own signal $s_n$.
Therefore, this claim holds for node $n$.

Suppose the claim holds for node $i+1, i+2,\dotsb, n$.
Then node $i$ either reveals its own signal, which is completely independent of $c^{i-1}$, or be ignored by later nodes (by the induction assumption).
In the later case, node $i$ will try to minimize its own failure probability by performing a Majority Algorithm on $\text{Valid}(c^{i-1})$ and $s_i$.
Therefore, the induction proof is completed.
With this claim, we can represent the optimal algorithms as:
\begin{equation*}
L_i
= \left\{
	\begin{array}{c l}
		s_i & c^{i-1} \in RS_i \\
 		\Maj[\text{Valid}(c^{i-1}), s_i] & c^{i-1} \not\in RS_i
 	\end{array}
	\right..
\end{equation*}

(2) The private signal can be deprived from the Majority function.

Assume the contrary that under some previous decision vector $c^{i-1}$ not in reveal set $RS_i$, node $i$'s signal $s_i$ affects $c_i$.
Then we have $c_i = L_i(c^{i-1},s_i) = s_i$.
But according to the definition, the input $c^{i-1}$ should belong to $RS_i$, which contradicts to our assumption..
Therefore, $L_i$ can be written in the form
\begin{equation*}
L_i
= \left\{
	\begin{array}{c l}
		s_i & c^{i-1} \in RS_i \\
 		\Maj[\text{Valid}(c^{i-1})] & c^{i-1} \not\in RS_i
 	\end{array}
	\right..
\end{equation*}

(3) If a node performs Majority Algorithm, all later nodes will also perform Majority Algorithm.

The intuition here is that the act of revealing one's own signal is beneficial to all later nodes.
Therefore, such ``benefits'' should be placed at the very front so that the positive influence is maximized.
We now consider a proof by contradiction.
Assume the contrary, then there must exist circumstances where node $i$ performs majority, but node $i + 1$ reveals his own signal.
In this case, we can modify the algorithm to let node $i$ reveal its signal and node $i+1$ perform majority.
It can be shown that the $\mathcal{E}$ after this modification will be no worse than the original one.
So there must exist algorithms $(L_1, \dots, L_n)$ with minimal $\mathcal{E}$ that satisfies this claim.
In such optimal algorithms, $\Maj[\text{Valid}(c^{i-1})]$ would be equivalent to $\Maj[c^{i-1}]$
	and $L_i$ may therefore be simplified to
\begin{equation*}
L_i
= \left\{
	\begin{array}{c l}
		s_i & c^{i-1} \in RS_i \\
 		\Maj[c^{i-1}] & c^{i-1} \not\in RS_i
 	\end{array}
	\right..
\end{equation*}
This completes our proof.
\end{proof}

\subsection{Efficient algorithms to derive $\delta_n(\cdot)$ for any $n$}
\label{subsec:algorithm}
Here, we provide algorithms that efficiently derive $\delta_n(i)$ in average $O(\log n)$ time.
Denote $\mathcal{E}(i,d)$ to be the expected number of wrong nodes given that $|\diff(c^{i-1})| = d$.
The idea is to use recursion to derive $\mathcal{E}(i,d)$ for all $i$ and $d$, in the process of which all $\delta_n(i)$ may be calculated.
First of all, it is already known that node $n$ outputs the majority of former outputs.
Thus $\mathcal{E}(n,d)$ may be efficiently calculated for all $d$.
Secondly, we show how to calculate $\{\mathcal{E}(k,d)|d\}$ given $\{\mathcal{E}(i,d)~|~d, i > k\}$.
\begin{itemize}
\item Given $|\diff(c^{i-1})| = d$, we can calculate the Bayesian probability of the nature bit as in Equation $\ref{equ:bayesian}$.
\item Node $k$ calculates the expected number of wrong nodes when it chooses to reveal private signal or do majority.
	This calculation is in $O(1)$ given $\mathcal{E}(k+1,d+1)$ and $\mathcal{E}(k+1,d-1)$.
	$\mathcal{E}(k,d)$ is set to be the smaller of the two calculated results.
\item Node $k$ performs Majority Algorithm under difference $d$ if and only if the $\mathcal{E}$ calculated for Majority Algorithm is smaller than that for revealing private signal.
\end{itemize}
In the second step, if node $k$ chooses to reveal its private signal, $\mathcal{E}(k, d)$ is updated as
\begin{equation}
\label{equ:revealupdate}
\begin{aligned}
	\mathcal{E}(k,d)
=
 	q_1\cdot \mathcal{E}(k+1, d+1) + (1-q_1)\cdot \mathcal{E}(k+1, d-1),
\end{aligned}
\end{equation}
where $q_1$ is the probability that node $k$'s private signal matches the majority of former outputs.
Similarly, if node $k$ chooses to do Majority Algorithm, $\mathcal{E}(k, d)$ is updated as
\begin{equation}
\label{equ:majorityupdate}
	\mathcal{E}(k,d)
=
	q_2\cdot (n - {k+d\over 2}) + (1-q_2)\cdot {k+d\over 2},
\end{equation}
where $q_2$ is the probability that the majority of former outputs is correct.
After running through all $k\in \{1, \dots, n\}$, we can find $\delta_n(i)$ as
\begin{equation}
\label{equ:deltaupdate}
	\delta_n(i)
	 =
	 \max_{d} \Big( \text{$L_i$ reveals private signal under input difference $d$} \Big).
\end{equation}
A detailed pseudo-code is presented as in \ref{alg:n3alg}.
\begin{algorithm}[htb]
\caption{$O(n^2)$ algorithm for finding $\{\delta_n(i)|i\}$} \label{alg:n3alg}
\begin{algorithmic}[1]
\hide{
\REQUIRE ~~\\
	The set of positive samples for current batch, $P_n$;\\
	The set of unlabeled samples for current batch, $U_n$;\\
	Ensemble of classifiers on former batches, $E_{n-1}$;\\
\ENSURE ~~\\
	Ensemble of classifiers on the current batch, $E_n$;
	}
\STATE Calculate $\mathcal{E}(n,d)$ for all $i\in [n]$ using Equation ($\ref{equ:majorityupdate}$).
\FOR {$i = n-1 \text{ to } 1$, $d = 1\to \text{ to } i$}
	\STATE Update $\mathcal{E}(i,d)$ with the smaller one of Equation ($\ref{equ:revealupdate}$) and Equation ($\ref{equ:majorityupdate}$);
	\STATE Let node $i$ reveals if and only if Equation ($\ref{equ:revealupdate}$) is larger than Equation ($\ref{equ:majorityupdate}$);
\ENDFOR
\FOR {$i = 1 \text{ to } n$}
	\STATE Update $\delta_n(i)$ using Equation ($\ref{equ:deltaupdate}$).
\ENDFOR
\RETURN the $\delta_n(\cdot)$ function;

\end{algorithmic}
\end{algorithm}
\\
Algorithm $\ref{alg:n3alg}$ runs in time $O(n^2)$, which is polynomial yet still improvable.
In Algorithm $\ref{alg:n3alg}$, we make no use of the properties of $\delta_n()$,
	some of which may be especially useful.
For example, we can show that for any $n$ and $i\in[n]$,
\begin{equation}
\label{equ:bound3}
\delta_n(i+1)-1\le \delta_n(i)\le \delta_n(i+1)+1.
\end{equation}
Equation $\ref{equ:bound3}$ can be proved through induction.
However, it is easier to understand it intuitively.
Two consecutive nodes own almost the same set of data and thus should share similar criteria.
This is why $\delta_n(i)$ and $\delta_n(i+1)$ should be close to each other.
By applying Equation $\eqref{equ:bound3}$ back to the algorithms,
	we can simply test the difference $d = \delta_n(i+1)+1$ for $L_i$,
	instead of testing all possible difference $d$.
The detailed pseudo-code for the new algorithm is presented as below.
\\
\begin{algorithm}[htb]
\caption{Calculate-$\mathcal{E}$}
\begin{algorithmic}[1]
\REQUIRE ~~\\
	The node's position, $i$;\\
	The difference witnessed in former outputs, $d$;
\ENSURE ~~\\
	$d < i$;
\IF {$\delta_n(i)$ has already been calculated}
	\IF {$d > \delta_n(i)$}
		\STATE Calculate $\mathcal{E}(i,d)$ using Equation \ref{equ:revealupdate}.
		\STATE If $\mathcal{E}(i+1, d+1)$ or $\mathcal{E}(i+1, d-1)$ hasn't been calculated before, recursively apply Calculate-$\mathcal{E}$();
	\ELSE
		\STATE Calculate $\mathcal{E}(i,d)$ using Equation \ref{equ:majorityupdate};
	\ENDIF
	\RETURN True (it doesn't matter what we return here);
\ELSE
	\STATE Update $\mathcal{E}(i,d)$ with the smaller one of Equation \ref{equ:revealupdate} and Equation \ref{equ:majorityupdate};
	\IF {Equation \ref{equ:revealupdate} is larger than Equation \ref{equ:majorityupdate}}
		\RETURN True;
	\ELSE
		\RETURN False;
	\ENDIF
\ENDIF

\end{algorithmic}
\end{algorithm}
\begin{algorithm}[htb]
\caption{A $O(n\log n)$ algorithm for finding $\{\delta_n(i)|i\}$} \label{alg:nlognalg}
\begin{algorithmic}[1]
\STATE Set $\delta_n(n) = n~ \% ~2$.
\FOR {$i = n-1 \text{ to } 1$}
	\STATE Call function Calculate-SW$(i,\delta_n(i+1)+1)$;
	\STATE If the output is true, set $\delta_n(i) = \delta_n(i+1)+1$; otherwise, set $\delta_n(i) = \delta_n(i+1)-1$.
\ENDFOR
\RETURN the $\delta_n(\cdot)$ function;
\end{algorithmic}
\end{algorithm}

In Algorithm \ref{alg:nlognalg}, we only need to calculate $\{\mathcal{E}(i,d) ~|~ d \le \delta_n(i) + 1\}$.
Since $\delta_n(0) = O(\log(n))$, only $O(n\log n)$ calls to Calculate-$\mathcal{E}$() is needed.
And if we perform some $O(n)$-time preprocessing, each call to Calculate-$\mathcal{E}$() can be finished within $O(1)$ time.
Therefore, Algorithm $\ref{alg:nlognalg}$ finishes in $O(n\log n)$ time.
Another thing learnt from this algorithm is that, $\delta_n(i)$ depends only on the number of nodes left, which is $n-i$.
So we can rewrite $\delta(i)$ as an universal function $\delta(\cdot)$ independent of $n$, where $\delta_n(i) = \delta(n-i)$.
Thus the output of Algorithm \ref{alg:nlognalg} not only determines the optimal parameters for $n$ nodes,
	but also implies the optimal parameters for all $m$-sized graph, where $m<n$.
This property would be highly useful in real life, where the number of nodes is flexible but controlled within a certain range.

\section{On the optimal topology of layer graphs}

\subsection{Optimality analysis for two-layer graphs} \label{subsec:two-layer optimal}
For the optimal topology among two layer graphs, its first layer has size $\log_{s}(n)-\log_s( \log_s n)/2 + O(1)$.
\begin{proof}
To solve for the optimal topology, we perform an optimization over Equation $\ref{equ:equ1}$
\begin{equation*}
	f(a_1) = \arg\min_{a_1} \Big( (1-p)\cdot a_1 + p_w(a_1)\cdot (n - a_1) \Big),
\end{equation*}
where $a_1$ is used to denote the size of the first layer.
For simplicity, we assume $a_1$ to be odd and is consequently denoted as $2k+1$.
Such relaxation results in only $O(1)$ error in the optimized $a_1$, and by Remark $\ref{rmk:estimation}$, can be allowed.
By definition, $p_w(2k+1)$ can be extended into the following form:
\begin{equation*}
	p_w(2k+1)
=
	\sum_{i=0}^{k-1} \binom{2k+1}{i} (1-p)^{2k+1-i} p^i.
\end{equation*}
Optimization over summation is usually hard.
So we first estimate variables such as $p_w(2k+1)$ and relax the equation with suitable operation.
By Chernoff Bound, we have,
\begin{equation*}
	1 - p_w(2k+1)
\geq
	1 - e^{-2(p(2k+1)-k+1)^2/(2k+1)}
\geq
 	1 - e^{-(2p-1)\cdot 2k}.
\end{equation*}
Also, when $n$ is large, by Stirling Approximation, we have
\begin{equation*}
	\binom{2k+1}{k}p^k(1-p)^k
\sim
	\frac{2k+1}{k+1} \cdot \frac{[4p(1-p)]^k}{\sqrt{\pi \cdot 2k}}
\sim
	0.
\end{equation*}
We now consider the equation $f(2k+3) - f(2k+1) = 0$.
If this equation can be shown to have unique solution,
	then that very solution would be the unique extreme point of $f(\cdot)$.
Subtract the definition of $f(2k+3)$ and $f(2k+1)$,
\begin{equation}
\label{equ:original}
 	f(2k+3) - f(2k+1)
=	
	2 - 2p + 2p_w(2k+1) + (n - 2k - 3)(p_w(2k+1) - p_w(2k+3)).
\end{equation}
By subtracting $p_w(2k+3)$ and $p_w(2k+1)$, we can get
\begin{equation}
\label{equ:subtract}
\begin{aligned}
& ~~
	p_w(2k+1) - p_w(2k+3)
\\
= &~~
	\binom{2k+1}{k+1}p^k(1-p)^{k+1}(1-p^2)
	+\binom{2k+1}{k} p^{k+1}(1-p)^k(1-p)^2
	-\binom{2k+1}{k+1} p^k(1-p)^{k+1}.
\end{aligned}
\end{equation}
Take Equation $\ref{equ:subtract}$ back into Equation $\ref{equ:original}$, we can achieve a further simplified equation that
\begin{equation*}
	2(1 - p)
 =
	2 p_w(2k+1) +
	(n-2k-3)\frac{2k+1}{k+1}\binom{2k}{k}\left[p(1-p)\right]^{k+1}(2p-1).
\end{equation*}
By Stirling Approximation, we have,
\begin{equation*}
	\binom{2k}{k}
=
	\frac{(2k)!}{k!k!}
=
	\frac{\lambda_{2k}}{\lambda_k^2}\cdot \frac{\sqrt{2\pi \cdot 2k}(2k)^{2k}}{2\pi k \cdot k^{2k}}
=
	\frac{\lambda_{2k}}{\lambda_k^2}\cdot \frac{2^{2k}}{\sqrt{\pi k}},
\end{equation*}
where $\lambda$ are the stirling coefficients such that,
\begin{equation*}
	e^{-\frac{19}{150}}
<
	\frac{e^\frac{1}{24k+1}}{e^\frac{1}{6k}}
\leq
	\frac{\lambda_{2k}}{\lambda_k^2}
\leq
	\frac{e^\frac{1}{24k}}{e^\frac{2}{12k+1}}
<
	1.
\end{equation*}
Let $c_0 = \lambda_{2k} / \lambda_{k^2}$, then prior equations imply that
\begin{equation}
\label{equ:jun}
	\frac{2(1-p)-2p_w(2k+1)}{c_0p(1-p)(2p-1)(2-\frac{1}{k+1})}
=
	(n-2k-3)[4p(1-p)]^k \frac{1}{\sqrt{\pi k}}.
\end{equation}
We denote the left side as $L(k)$ and right side as $R(k)$. Since
\begin{equation*}
	\lim_{k\to +\infty} p_w(2k+1)
=
	0.
\end{equation*}
So there exists a threshold $K_0$ s.t. $ \forall k > K_0 $ we have $p_w(2k+1) < (1-p)/2$.
Thus when $k$ is large, we have 2 constant bounds $ C_1 $, $ C_2 $ s.t.
\begin{equation}
\label{equ:bound of equation}
	C_1
~=~
	\frac{1-p}{p(1-p)(2p-1)\cdot 2}
~<~
	L(k)
~<~
	\frac{2(1-p)}{p(1-p)(2p-1)\frac{3}{2}\cdot e^{-\frac{19}{150}}}
=
	C_2.
\end{equation}
This bounds the left hand size of $\ref{equ:jun}$ into constant range.
Now consider the right hand size $R(k)$.
It can be noticed that:
\begin{itemize}
\item $ n-2k-3 $ strictly decrease with $ k $.
\item $ p(1-p) < \frac{1}{4} \Rightarrow 4p(1-p)<1 $, so $ \left[4p(1-p)\right]^k $ strictly decrease with $ k $.
\item $ {1}\Big/{\sqrt{\pi k}} $ decrease with $ k $.
\item $R(k)$ is the multiplication of the three items above.
\end{itemize}
Therefore, $R(k)$ strictly decreases with $k$.
By its monotonicity, we can have an inverted function $ R^{-1}(\cdot)$ to bound the value of $k$ using $C_1$ and $C_2$,
	such that $R^{-1}(C_2) < k < R^{-1}(C_1)$.
Suppose $R(k)=C$ for some C, then we can derive
\begin{equation*}
	k
=
	\frac{\log(n-2k-3)-\frac{1}{2}\log(\pi k) - \log C}{-\log(4p(1-p))}
=
	\log_{s}(n)-{\log_s( \log_s n)\over 2} + O(1).
\end{equation*}
This completes our proof.
\end{proof}

\subsection {Optimality analysis for general layer graphs}
\label {subsec:optimality all layers}

To solve for the optimal layer topology, we need to find the optimization of Equation $\ref{equ:recursive}$,
\begin{equation*}
\mathcal{E}(a_1, \dotsb, a_k)
=
(1-p)\cdot a_1 + p_w(a_1)\cdot (n - a_1) + p_n(a_1)\cdot \mathcal{E}(a_2, \dotsb, a_k),
\end{equation*}
For simplicity, denote $g(k, n)$ as the optimal loss of a layer graph whose nodes' number is $n$ and whose first layer has size $k$.
An equivalent version of Equation $\ref{equ:recursive}$ is
\begin{equation*}
g(k, n) = (1-p)k + p_w(k)(n-k) + p_n(k)\min_{k'}f(k',n-k).
\end{equation*}
We already know that the optimal loss of layer graphs can be tightly approximated by $c\log n + o(\log n)$ for some constant $c$.
Therefore, if we apply similar method when finding the optimal two-layer topology and subtract $g(2k+1, n)$ from $g(2k+3, n)$,
	we can get
\begin{align*}
	g(2k+3,n) - g(2k+1,n)
& =
	f(2k+3) - f(2k+1) + c\log {n - 2k - 3 \over n - 2k - 1}\cdot p_n(2k+1) ) \\
& =
	f(2k+3) - f(2k+1) + O( \log (1 - {2 \over n - 2k - 1})\cdot (4p(1-p))^{k}) \\
& =
	f(2k+3) - f(2k+1) + O(n^{-2}),
\end{align*}
where $f$ follows from the definition in Equation $\ref{equ:equ1}$.
This implies that the optimization of the first layer's size in general layer graphs can be closely approximated by
	the first layer size in the optimal two-layer graph.
Therefore, $p_n(a_1)\mathcal{E}(a_2, \dotsb, a_k)$ can be safely discarded from Equation $\ref{equ:recursive}$.
And in the optimal topology, the first layer should also have $\log_{s} n - o(\log n)$ nodes.

Using similar method, the second layer should have size $\log_s ( n - \log_s n) \sim \log_s n$.
We can apply this argument repeatedly, which eventually leads to the results in Theorem $\ref{thm:optimal layer topology}$.





%
%
%
%
%
%
%
%

\end{document}